\documentclass[conference]{IEEEtran}
\usepackage{ifpdf}

\ifCLASSINFOpdf
  \usepackage[pdftex]{graphicx}
\else
\fi
\usepackage{amsmath}
\usepackage{amsfonts}
\usepackage{amssymb}
\usepackage{amsthm}
\usepackage{algorithmic}

\usepackage{array}
\usepackage{float}
\usepackage{url}

\newtheorem{theorem}{Theorem}
\newtheorem{definition}{Definition}
\newtheorem{example}{Example}

\newtheorem{corollary}{Corollary}
\newtheorem{lemma}{Lemma}
\newtheorem{remark}{Remark}

\begin{document}
%
\title{Construction $C^\star:$ an inter-level coded version of Construction C}

\author{\IEEEauthorblockN{Maiara F. Bollauf \IEEEauthorrefmark{1},
Ram Zamir \IEEEauthorrefmark{2} and
Sueli I. R. Costa \IEEEauthorrefmark{3}}
\vspace{0.2cm}
\IEEEauthorblockA{\IEEEauthorrefmark{1} \IEEEauthorrefmark{3} Institute of Mathematics, Statistic and Computer Science\\
University of Campinas, Sao Paulo, Brazil\\ Email: maiarabollauf@ime.unicamp.br, sueli@ime.unicamp.br}
\vspace{0.2cm}
\IEEEauthorblockA{\IEEEauthorrefmark{2} Deptartment Electrical Engineering-Systems \\
Tel Aviv University, Tel Aviv, Israel \\ Email: zamir@eng.tau.ac.il}}

\maketitle

\begin{abstract} 

	Besides all the attention given to lattice constructions, it is common to find some very interesting nonlattice constellations, as Construction C, for example, which also has relevant applications in communication problems (multi-level coding, multi-stage decoding, good quantization efficieny).  In this work we present a constellation which is a subset of Construction C, based on inter-level coding, which we call Construction $C^\star.$ This construction may have better immunity to noise and it also provides a simple way of describing the Leech lattice $\Lambda_{24}.$  A condition under which Construction $C^{\star}$ is a lattice constellation is given.
	
	
\end{abstract}

{\small \textbf{\textit{Index terms}---Lattice construction, Bit-interleaved coded modulation (BICM), Construction $C^\star$, Construction by Code-Formula, Leech lattice.}}


%
\IEEEpeerreviewmaketitle

\section{Introduction}

		Communication problems involve, in general, transmitting digital information over a channel with minimum losses. One way to approach it is by using coded modulation \cite{debuda1}, where not only coding, but also mapping the code bits to constellation symbols is significant. In the latest years, a prevalent coded modulation scheme is the bit-interleaved coded modulation (BICM), which is the motivation to our study.
		
		BICM, first introduced by Zehavi \cite{zehavi,caire}, asks mainly to have: an $nL-$dimensional binary code $\mathcal{C},$ an interleaver (permutation) $\alpha$ and a one-to-one binary labeling map $\mu: \{0,1\}^{L} \rightarrow \mathcal{X},$ where $\mathcal{X}$ is a signal set $\mathcal{X}=\{0,1, \dots, 2^{L}-1\}$ in order to construct a constellation $\Gamma_{BICM}$ in $\mathcal{X}^{n} \subseteq \mathbb{R}^{n}.$ The code and interleaveled bit sequence $c$ is partitioned into $L$ subsequences $c_i$ of length $n:$
\begin{equation}
c=(c_{1}, \dots, c_{L}), \ \ \mbox{with} \ \ c_{i}=(c_{i1},c_{i2}, \dots, c_{in}). 
\end{equation}

	The bits $c_{1j}, \dots, c_{Lj}$ are mapped at a time index $j$ to a symbol $x_{i}$ chosen from the $2^{L}-$ary signal constellation $\mathcal{X}$ according to the binary labeling map $\mu.$ Hence, for a $nL-$binary code $\mathcal{C}$ to encode all bits, then we have the scheme  below:
	
{\small \begin{center}

$\boxed{\text{codeword} \ (c)}$ $\rightarrow$ $\boxed{\text{interleaver} \ \alpha }$ $\rightarrow$ $\boxed{\text{partitioning into} \ $L$ \ \text{subsequences of length} \ $n$}$ $\rightarrow$ $\boxed{\text{mapping} \ \mu}$ $\rightarrow$ $\boxed{x_{j}=\mu(c_{1j}, \dots, c_{Lj}), \ j={1, \dots, n}}$

\end{center}}
	
	Under the natural labeling $\mu(c_{1}, c_{2},, \dots, c_{L})=c_{1}+2c_{2}+\dots+2^{L-1}c_{L}$ and assuming identity interleaver $\alpha(\mathcal{C})=\mathcal{C},$ it is possible to define an extended BICM constellation in a way very similar to the well known multilevel Construction C, that we call Construction $C^{\star}.$
	
	The constellation produced via Construction $C^\star$ is always a subset of the constellation produced via Construction C for the same projection codes (as defined below) and it also does not usually produce a lattice. The objective of our paper is to explore this new construction, aiming to find a condition that makes it a lattice and also to describe the Leech lattice $\Lambda_{24}$ with Construction $C^\star.$
	
	The paper is organized as follows: Section II shows some preliminary definitions; in Section III we introduce Construction $C^{\star},$ illustrate it with examples and also show how to describe the Leech lattice using this construction; in Section IV we exhibit a condition for $\Gamma_{C^{\star}}$ to be a lattice and Section V is devoted to conclusions.


 

\section{Mathematical background}
	
	In this section, we will introduce the basic concepts, notation and results to be used in the sequel. We will denote by $+$ the real addition and by $\oplus$ the sum in $\mathbb{F}_{2},$ i.e., $x \oplus y=(x+y)mod \ 2.$
	
%
%

\begin{definition}(Lattice) A lattice $\Lambda \subset \mathbb{R}^{N}$ is a set of integer linear combinations of independent vectors $v_{1}, v_{2}, \dots, v_{n} \in \mathbb{R}^{N},$ with $n \leq N.$ 
\end{definition}
	
	It is possible to derive lattice constellations from linear codes using the known Constructions $A$ and $D$ \cite{conwaysloane}. 
	
	\begin{definition}(Construction A) Let $\mathcal{C}$ be a linear $(n,k,d)-$binary code. We define the binary Construction A as
\begin{equation}
\Lambda_{A}=\mathcal{C} + 2\mathbb{Z}^{n}.
\end{equation}
\end{definition}

\begin{definition} (Construction D)  Let $\mathcal{C}_{1} \subseteq \dots \subseteq \mathcal{C}_{L} \subseteq \mathbb{F}_{2}^{n}$ be a family of nested linear binary codes. 
	Let $k_{i}=\dim(\mathcal{C}_{i})$ and let $b_{1}, b_{2}, \dots, b_{n}$ be a basis of $\mathbb{F}_{2}^{n}$ such that $b_{1}, \dots, b_{k_{i}}$ span $\mathcal{C}_{i}.$ The lattice $\Lambda_D$ consists of all vectors of the form
	\begin{equation}
	\displaystyle\sum_{i=1}^{L} 2^{i-1} \displaystyle\sum_{j=1}^{k_{i}} \alpha_{ij} b_{j}+2^{L}z
	\end{equation}
	where $\alpha_{ij} \in \{0,1\}$ and $z \in \mathbb{Z}^{n}.$ 
\end{definition}

	Another remarkable and well studied multi-level construction, that in general does not produce a lattice constellation, even when the underlying codes are linear, is Construction C, defined below using the terminology in \cite{forney1} (more details and applications also in \cite{agrelleriksson} \cite{bollaufzamir}).
	
\begin{definition}(\textit{Construction C})  Consider $L$ binary codes $\mathcal{C}_{1}, \dots, \mathcal{C}_{L} \subseteq \mathbb{F}_{2}^{n},$ not necessarily nested or linear. Then we define an infinite constellation $\Gamma_{C}$ in $\mathbb{R}^{n}$ that is called Construction C as:
	\begin{equation} \label{eqC}
	\Gamma_{C}:=\mathcal{C}_{1}+2\mathcal{C}_{2}+ \dots + 2^{L-1}\mathcal{C}_{L}+2^{L}\mathbb{Z}^{n},
	\end{equation} 
	or equivalently
	{\small \begin{eqnarray}
	\Gamma_{C} &:=& \{c_1 + 2c_2 + \dots + 2^{L-1} c_L + 2^L z: c_i \in \mathcal{C}_i,  \nonumber \\
	& & i=1,\dots,L, \ z \in \mathbb{Z}^n\}.
	\end{eqnarray}}
\end{definition}

	Note that if $L=1$ and we consider a single level with a linear code, then both Constructions C and D reduce to lattice Construction A. There exists also a relation between Constructions C and D that will be presented in what follows.
		
\begin{definition} (\textit{Schur product}) For $x=(x_{1}, \dots, x_{n})$ and $y=(y_{1}, \dots, y_{n}) \in \mathbb{F}_{2}^{n},$ we define $x \ast y = (x_{1}y_{1}, \dots, x_{n}y_{n}).$
\end{definition}

	It is easy to verify, using the Schur product that for $x, y \in \mathbb{F}_{2}^{n}$
\begin{equation} \label{sum}
x+y=x \oplus y + 2(x \ast y).
\end{equation}

	Denote by $\Lambda_{C}$ the smallest lattice that contains $\Gamma_{C}.$ Kositwattanarerk and Oggier \cite{kositoggier} give a condition that if satisfied guarantees that Construction C will provide a lattice which coincides with Construction D.

\begin{theorem} \label{thmko} \cite{kositoggier} (\textit{Lattice condition for Constructions C and D}) Given a family of nested linear binary codes $\mathcal{C}_{1} \subseteq \dots \subseteq \mathcal{C}_{L} \subseteq \mathbb{F}_{2}^{n},$ then the following statements are equivalent:
	\begin{itemize}[\IEEEsetlabelwidth{Z}]
		\item[1.] $\Gamma_{C}$ is a lattice.
		\item[2.] $\Gamma_{C}=\Lambda_{C}.$
		\item[3.] $\mathcal{C}_{1} \subseteq \dots \subseteq \mathcal{C}_{L} \subseteq \mathbb{F}_{2}^{n}$ is closed under Schur product, i.e., given two elements $c_{i}, \tilde{c}_{i} \in \mathcal{C}_{i}, c_{i} \ast \tilde{c}_{i} \in \mathcal{C}_{i+1},$ for all $i=1, \dots, L-1.$
		\item[4.] $\Gamma_{C}=\Lambda_{D},$ 
	\end{itemize}
\end{theorem} 

\section{Construction $C^\star$ over binary codes}

	This section is devoted to the introduction of a new method of constructing constellations from binary codes: Construction $C^\star.$


\begin{definition} \label{constrcstar} (\textit{Construction $C^{\star}$}) Let $\mathcal{C}$ be an $nL-$dimensional code in $\mathbb{F}_{2}^{nL}.$ Then Construction $C^{\star} \in \mathbb{R}^{n}$ is defined as 
{\small \begin{eqnarray}
 \Gamma_{C^{\star}}& := & \{c_{1}+2c_{2}+ \dots + 2^{L-1}c_{L}+2^{L}z: (c_{1}, c_{2}, \dots, c_{L}) \in \mathcal{C}, \nonumber \\
& & c_{i} \in \mathbb{F}_{2}^{n}, i=1, \dots, L, z \in \mathbb{Z}^{n}\}.  
\end{eqnarray}}
\end{definition}

\begin{definition} (\textit{Projection codes}) \label{subcodes} Let $(c_1,c_2,...,c_L)$ be a partition of a codeword $c = (b_1,....,b_{nL}) \in \mathcal{C}$ into length$-n$ subvectors  $c_i = (b_{(i-1)n+1},....,b_{in}),$  $i=1,\dots,L.$ Then, a projection code $\mathcal{C}_i$ consists of all vectors $c_{i}$ that appear as we scan through all possible codewords $c \in \mathcal{C}.$ Note that if $\mathcal{C}$ is linear, every projection code $\mathcal{C}_{i}, i=1, \dots, L$ is also linear.
\end{definition}



\begin{remark} If $\mathcal{C}=\mathcal{C}_{1} \times \mathcal{C}_{2} \times \dots \times \mathcal{C}_{L}$ then Construction $C^{\star}$ coincides with Construction C, because the projection codes are independent. However, in general, the projection codes are dependent, i.e., not all combinations compose a codeword in the main code $\mathcal{C}$ so we get a subset of Construction C., i.e., $\Gamma_{C^\star} \subseteq \Gamma_{C}.$
\end{remark}

\begin{definition}(\textit{Associated Construction C}) Given a Construction $C^{\star}$ defined by a linear binary code $\mathcal{C} \subseteq \mathbb{F}_{2}^{nL},$ we call the associated Construction C the constellation defined as
\begin{equation}
\Gamma_{{C}}= \mathcal{C}_{1} + 2 \mathcal{C}_{2} + \dots + 2^{L-1}\mathcal{C}_{L}+2\mathbb{Z}^{n},
\end{equation}
such that $\mathcal{C}_{1}, \mathcal{C}_{2}, \dots, \mathcal{C}_{L} \in \mathbb{F}_{2}^{n}$ are the projection codes of $\mathcal{C}$ as in Definition \ref{subcodes}.

\end{definition}


	One can observe that the immediate advantage of working with Construction $C^{\star}$ instead of Construction C lies in the fact that a code of block length $nL$ typically has a larger minimum Hamming distance and may present a better immunity to noise than a code of block length $n.$ 
	

\begin{example} Consider a linear binary code $\mathcal{C}$ with length $nL=4,$ ($L=n=2$), where $\mathcal{C}=\{(0,0,0,0), (1,0,0,1),(1,0,1,0),$ $(0,0,1,1)\} \subseteq \mathbb{F}_{2}^{4}.$ 
	Thus, an element $w \in \Gamma_{C^{\star}}$ can be written as
\begin{equation} \label{eqconstrc}
w=c_{1}+2c_{2}+4z \ \in \Gamma_{C^\star} ,
\end{equation}	
such that $(c_{1},c_{2}) \in \mathcal{C}$ and $z \in \mathbb{Z}^{2}.$ Geometrically, the resulting constellation is given by the blue points represented in Figure \ref{consc}. Note that $\Gamma_{C^{\star}}$ is not a lattice because, for example, $(1,2), (3,0) \in \Gamma_{C^{\star}},$ but $(1,2)+(3,0)=(4,2) \notin \Gamma_{C^{\star}}.$ However, if we consider the associated Construction C with codes $\mathcal{C}_{1}=\{(0,0),(1,0)\}$ and $\mathcal{C}_{2}=\{(0,0),(1,1),(0,1),(1,0)\},$ we have a lattice (Figure \ref{consc}), because $\mathcal{C}_{1}$ and $\mathcal{C}_{2}$ satisfy the condition given by Theorem \ref{thmko}.
	

\begin{figure}[H]
\begin{center}
		\includegraphics[height=5cm]{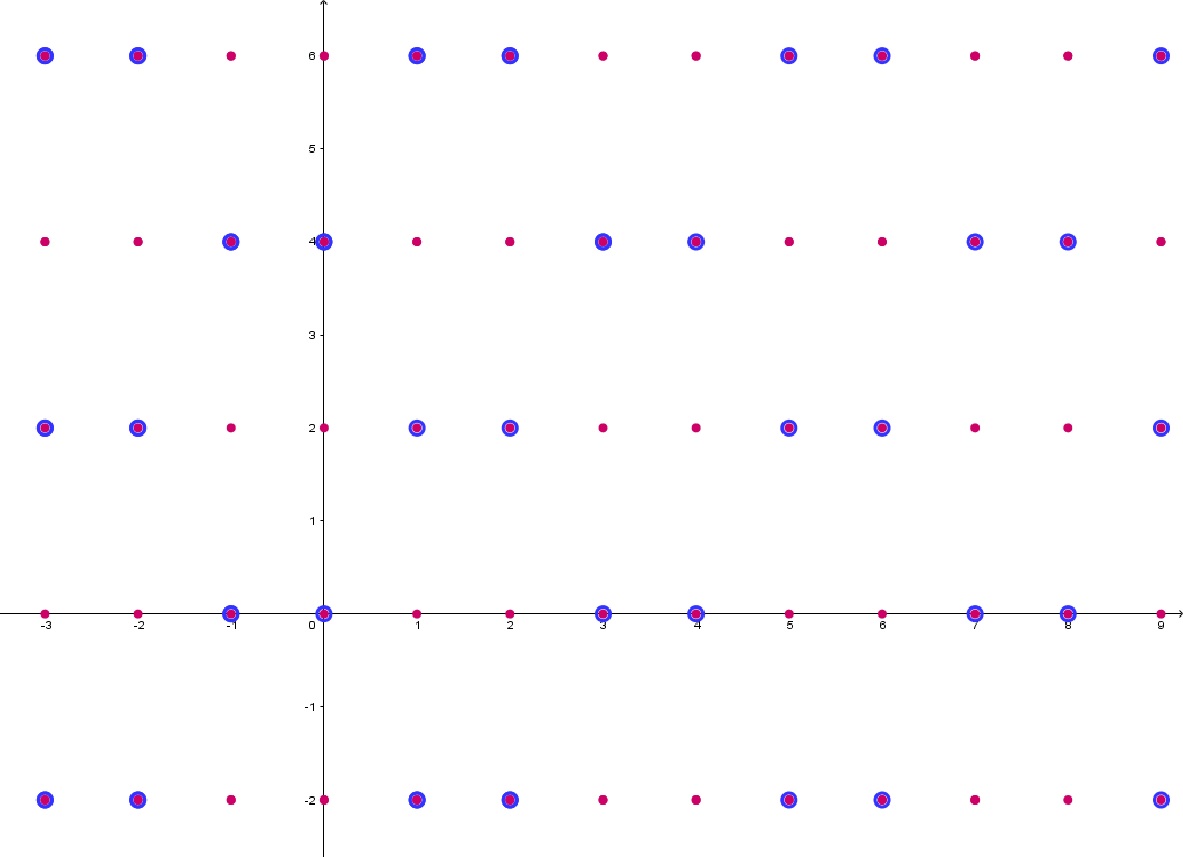}  
\caption{{(Nonlattice) Construction $C^{\star}$ constellation in blue and its associated (lattice) Construction C constellation in pink}}
 \label{consc}
\end{center}
\end{figure} 	
\end{example}

	The next example presents a case where both Constructions $C^\star$ and $C$ are lattices, but they are not equal.
	
\begin{example} \label{ceconjecture} Let a linear binary code $\mathcal{C}=\{(0,0,0,0),$ $(0,0,1,0),(1,0,0,1),(1,0,1,1)\} \subseteq \mathbb{F}_{2}^{4}$ ($nL=4,$ $L=n=2$), so the projection codes are $\mathcal{C}_{1}=\{(0,0),(1,0)\}$ and $\mathcal{C}_{2}=\{(0,0),(1,0),$ $(0,1),(1,1)\}.$ An element $w \in \Gamma_{C^{\star}}$ can be described as 
\begin{equation}
w=\begin{cases}
(0,0)+4z, & \mbox{if} \ {c}_{1}=(0,0) \ \mbox{and} \ {c}_{2}=(0,0) \\
(1,2)+4z, & \mbox{if} \ {c}_{1}=(1,0) \ \mbox{and} \ {c}_{2}=(0,1) \\
(2,0)+4z, & \mbox{if} \ {c}_{1}=(0,0) \ \mbox{and} \ {c}_{2}=(1,0) \\
(3,2)+4z, & \mbox{if} \ {c}_{1}=(1,0) \ \mbox{and} \ {c}_{2}=(1,1), \\
\end{cases}
\end{equation}
$z \in \mathbb{Z}^{2}.$ This construction is represented by black points in Figure \ref{fig10}. Note that $\Gamma_{C^{\star}}$ is a lattice and $\mathcal{C} \neq \mathcal{C}_{1} \times \mathcal{C}_{2},$ what implies that $\Gamma_{C^{\star}} \varsubsetneq \Gamma_{C}.$ However, the associated Construction C is also a lattice (Figure \ref{fig10}).

\begin{figure}[H]
\begin{center}
		\includegraphics[height=5cm]{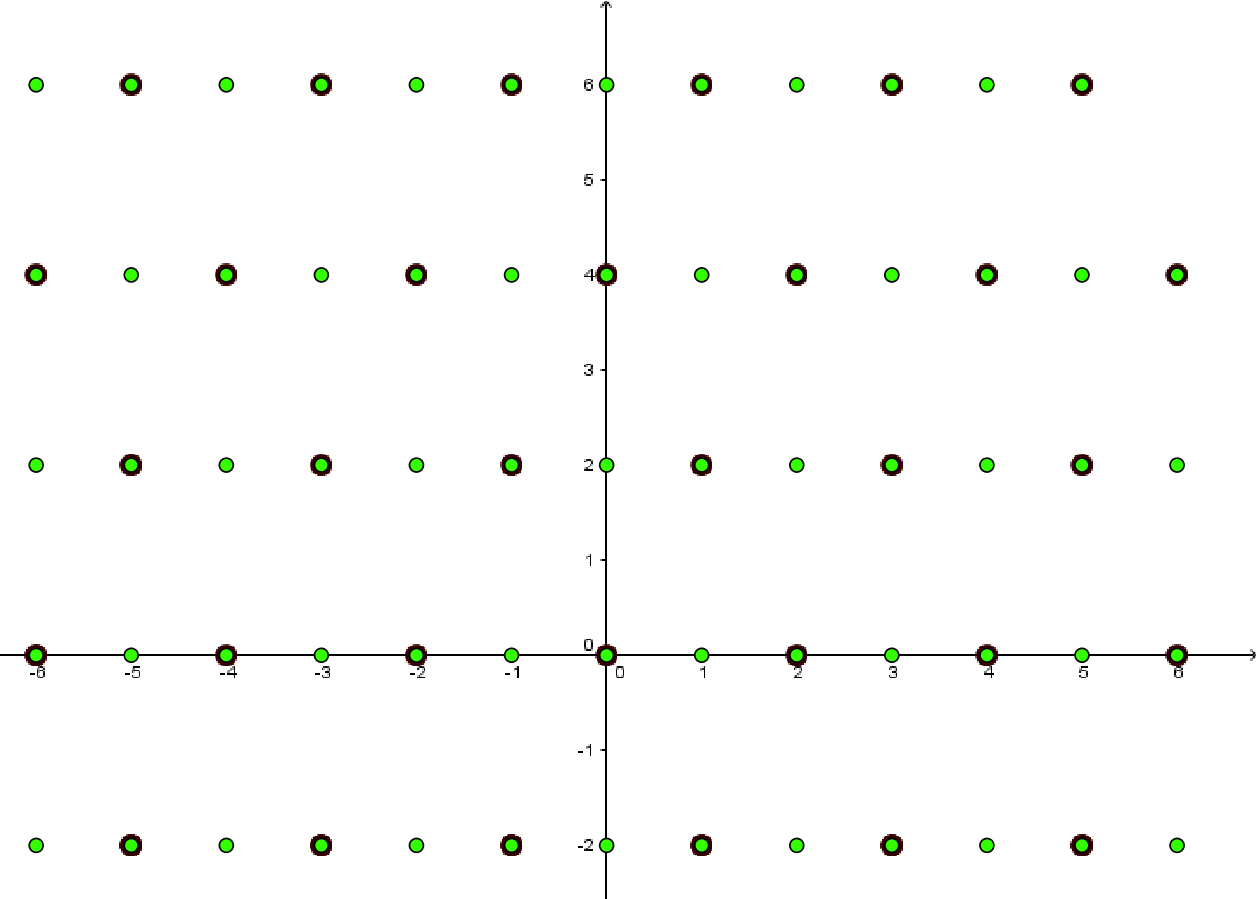}  
\caption{{(Lattice) Construction $C^{\star}$ constellation in black and its associated (lattice) Construction C constellation in green}}
 \label{fig10}
\end{center}
\end{figure} 	

\noindent To appreciate the advantage of $\Gamma_{C^\star}$ over the associated $\Gamma_{{C}},$ one can notice that the packing densities are, respectively $\Delta_{\Gamma_{C^\star}} = \frac{\Pi}{4} \approx 0.7853$ and  $\Delta_{\Gamma_{{C}}} = \frac{\Pi}{8} \approx 0.3926.$ Therefore, in this example, $\Gamma_{C^\star}$ has a higher packing density than  $\Gamma_{{C}}.$ 
\end{example}

	We can also describe the densest lattice in dimension $24$, the Leech lattice $\Lambda_{24},$ in terms of Construction $C^{\star}$ constellation with $L=3$ levels. 

\begin{example} \label{exleech} Based on the construction given by Conway and Sloane \cite{conwaysloane} (pp. 131-132) and Amrani et al \cite{amrani}, we start by considering three special linear binary codes
\vspace{0.1cm}
\begin{itemize}
\item $\mathcal{C}_{1}=\{(0,\dots, 0), (1, \dots, 1)\} \subseteq \mathbb{F}_{2}^{24};$
\vspace{0.2cm}
\item $\mathcal{C}_{2}$ as a Golay code $\mathcal{C}_{24} \subset \mathbb{F}_{2}^{24}$ achieved by adding a parity bit to the original $[23,12,7]-$binary Golay code $\mathcal{C}_{23},$ which consists in a quadratic residue code of length 23;
\vspace{0.2cm}
\item $\mathcal{C}_{3} = \tilde{\mathcal{C}}_{3} \cup \overline{\mathcal{C}}_{3} = \mathbb{F}_{2}^{24},$ where $\tilde{\mathcal{C}}_{3}=\{(x_{1}, \dots, x_{24}) \in \mathbb{F}_{2}^{24}:  \sum_{i=1}^{24} x_{1} \equiv 0 \mod 2\}$ and $\overline{\mathcal{C}}_{3}=\{(y_{1}, \dots, y_{24}) \in \mathbb{F}_{2}^{24}: \sum_{i=1}^{24} y_{1} \equiv 1 \mod 2\}.$
\end{itemize}

	Observe that $\mathcal{C}_{1}, \mathcal{C}_{2}$ and $\mathcal{C}_{3}$ are linear codes. Consider a code $\mathcal{C} \subseteq \mathbb{F}_{2}^{72}$ whose codewords are described in one of two possible ways:
\begin{eqnarray}
\mathcal{C}=\{ (0,\dots, 0, \underbrace{a_{1}, \dots, a_{24}}_{\in \mathcal{C}_{24}}, \underbrace{x_{1}, \dots, x_{24}}_{\in \tilde{\mathcal{C}}_{3}}), \nonumber \\
(1, \dots, 1, \underbrace{a_{1}, \dots, a_{24}}_{\in \mathcal{C}_{24}}, \underbrace{y_{1}, \dots, y_{24}}_{\in \overline{\mathcal{C}}_{3}})\}.
\end{eqnarray}

	Thus, we can define the Leech lattice $\Lambda_{24}$ as a $3-$level Construction $C^{\star}$ given by
{\small \begin{equation}
\Lambda_{24}=\Gamma_{C^{\star}}=\{c_{1}+2c_{2}+4c_{3}+8z: (c_{1}, c_{2}, c_{3}) \in \mathcal{C}, z \in \mathbb{Z}^{24}\}.
\end{equation}}
Observe that in this case $\Gamma_{C^{\star}} \neq \Gamma_{C}.$

In this case, the associated Construction C has packing density $\Delta_{\Gamma_{C}} \approx 0.00012 < 0.001929 \approx \Delta_{\Gamma_{C^\star}},$ which is the packing density of $\Lambda_{24},$ the best known packing density in dimension $24$ \cite{conwaysloane}.
	
\end{example}

\section{Conditions for latticeness of Construction $\mathcal{C}^\star$}


	In general, it is possible to have a lattice $\Gamma_{C^{\star}},$ with $\Gamma_{C^\star} \varsubsetneq \Gamma_{C},$ as can be observed in Example \ref{ceconjecture}. This fact motivated our search for a condition for a lattice Construction $C^{\star}.$ In the upcoming discussion, we will exhibit some definitions and present a condition for $\Gamma_{C^{\star}}$ to be a lattice.
	
		
\begin{definition} (\textit{Antiprojection}) The antiprojection (inverse image of a projection) $\mathcal
{S}_{i}(c_1,\dots, c_{i-1},c_{i+1}, \dots, c_{L})$ consists of all vectors $c_{i} \in \mathcal{C}_{i}$ that appear as we scan through all possible codewords $c \in \mathcal{C},$ while keeping $c_{1}, \dots, c_{i-1}, c_{i+1}, \dots, c_{L}$ fixed:
\begin{eqnarray}
&\mathcal{S}_{i}(c_1,...,c_{i-1}, c_{i+1},...,c_{L})= \nonumber \\
&\{c_{i} \in \mathcal{C}_{i}: (c_{1}, \dots, \underbrace{c_{i}}_{\text{\makebox[0pt]{i-th posititon} }}, \dots, c_{L}) \in \mathcal{C}\}.
\end{eqnarray}
\end{definition}

%





	The main contribution of this paper is the following:

\begin{theorem} (\textit{Lattice conditions for $\Gamma_{C^\star}$}) \label{thm} Let $\mathcal{C} \subseteq \mathbb{F}_{2}^{nL}$ be a linear binary code with projection codes $\mathcal{C}_{1},\mathcal{C}_{2}, \dots, \mathcal{C}_{L}$ such that $\mathcal{C}_{1} \subseteq \mathcal{S}_{2}(0,\dots,0) \subseteq \dots \subseteq \mathcal{C}_{L-1} \subseteq \mathcal{S}_{L}(0, \dots,0) \subseteq \mathcal{C}_{L} \subseteq \mathbb{F}_{2}^{n}.$ Then the constellation given by $\Gamma_{C^{\star}}$ represents a lattice if and only if $\mathcal{S}_{i}(0, \dots, 0)$ closes $\mathcal{C}_{i-1}$ under Schur product for all levels $i=2, \dots, L.$  
\end{theorem}

	The proof of Theorem \ref{thm} is given below, after a few motivational examples and related results.

	While $\mathcal{S}_{i}(0,\dots,0) \subseteq \mathcal{C}_{i}$ by construction, note that the assumption that $\mathcal{C}_{i} \subseteq \mathcal{S}_{i+1}(0,\dots,0),$ for $i=2, \dots, L,$ is not always satisfied by a general Construction $C^{\star},$ sometimes even if this Construction $C^\star$ is a lattice; see Example \ref{exnonesting}.

	Observe that when $\Gamma_{C^\star}=\Gamma_{C},$ i.e., when $\mathcal{C}=\mathcal{C}_{1} \times \mathcal{C}_{2} \times \dots \times \mathcal{C}_{L},$ we have that $\mathcal{S}_{i}(0,\dots, 0) = \mathcal{C}_{i}, i=1, \dots, L$ and our condition will coincide with the one presented in Theorem \ref{thmko}. Besides, if $\Gamma_{C^\star} \varsubsetneq \Gamma_{C}$ we also have that:

\begin{corollary}(\textit{Latticeness of associated Construction C}) Let $\mathcal{C} \subseteq \mathbb{F}_{2}^{nL}.$ If $\mathcal{C}_{1} \subseteq \mathcal{S}_{2}(0,\dots,0) \subseteq \dots \subseteq \mathcal{C}_{L-1} \subseteq \mathcal{S}_{L}(0, \dots,0) \subseteq \mathcal{C}_{L}$ and the constellation $\Gamma_{C^\star}$ is a lattice then also the associated Construction C is a lattice.
\end{corollary}

\begin{proof} If $\Gamma_{C^\star}$ is a lattice, conditions presented in Theorem \ref{thm} holds. 
For associated Construction C, $\mathcal{S}_{i}(0,\dots,0) = \mathcal{C}_{i},$ for all $i=1,\dots,L.$ Thus, it follows that $\mathcal{C}_{1} \subseteq \mathcal{C}_{2} \subseteq \dots \subseteq \mathcal{C}_{L}$ is closed under Schur product and $\Gamma_{{C}}$ is a lattice.
\end{proof}

	Before presenting the proof of Theorem \ref{thm}, we will see that the Leech lattice construction described in Example \ref{exleech} satisfies its condition.

\begin{example} 
We want to examine whether the proposed codes $\mathcal{C}_{1}, \mathcal{C}_{2}$ and $\mathcal{C}_{3}$ in Example \ref{exleech} satisfy the conditions stated by Theorem \ref{thm}.

	Observe that for these codes $\mathcal{S}_{2}(0,\dots,0)= \mathcal{C}_{2}$ and  $\mathcal{S}_{3}(0,\dots,0)= \tilde{\mathcal{C}}_{3}=\{(x_{1}, \dots, x_{24}) \in \mathbb{F}_{2}^{24}:  \sum_{i=1}^{24} x_{1} \equiv 0 \mod 2\}.$ Hence we need to verify that $\mathcal{C}_{1} \subseteq \mathcal{S}_{2}(0,\dots,0) \subseteq \mathcal{C}_{2}  \subseteq  \mathcal{S}_{3}(0,\dots,0) \subseteq \mathcal{C}_{3}$ and that $\mathcal{S}_{i}(0,\dots,0)$ closes $\mathcal{C}_{i-1}$ under Schur product for $i=2,3.$ 
	
	Indeed $\mathcal{C}_{1} \subseteq \mathcal{S}_{2}(0,\dots,0) = \mathcal{C}_{2},$ since $(0,\dots, 0) \in \mathcal{C}_{2}$ and if we consider the parity check matrix $H \in \mathbb{F}_{2}^{12 \times 24}$ of the $[24,12,8]-$Golay code
\begin{equation}
H=\begin{pmatrix}
B_{12 \times 12} & \mid & I_{12 \times 12}
\end{pmatrix},
\end{equation}
where 
\begin{equation}
B_{12 \times 12}=\left(
\begin{array}{cccccccccccc}
 1 & 1 & 0 & 1 & 1 & 1 & 0 & 0 & 0 & 1 & 0 & 1  \\
 1 & 0 & 1 & 1 & 1 & 0 & 0 & 0 & 1 & 0 & 1 & 1  \\
 0 & 1 & 1 & 1 & 0 & 0 & 0 & 1 & 0 & 1 & 1 & 1 \\
 1 & 1 & 1 & 0 & 0 & 0 & 1 & 0 & 1 & 1 & 0 & 1  \\
 1 & 1 & 0 & 0 & 0 & 1 & 0 & 1 & 1 & 0 & 1 & 1  \\
 1 & 0 & 0 & 0 & 1 & 0 & 1 & 1 & 0 & 1 & 1 & 1  \\
 0 & 0 & 0 & 1 & 0 & 1 & 1 & 0 & 1 & 1 & 1 & 1  \\
 0 & 0 & 1 & 0 & 1 & 1 & 0 & 1 & 1 & 1 & 0 & 1  \\
 0 & 1 & 0 & 1 & 1 & 0 & 1 & 1 & 1 & 0 & 0 & 1  \\
 1 & 0 & 1 & 1 & 0 & 1 & 1 & 1 & 0 & 0 & 0 & 1  \\
 0 & 1 & 1 & 0 & 1 & 1 & 1 & 0 & 0 & 0 & 1 & 1 \\
 1 & 1 & 1 & 1 & 1 & 1 & 1 & 1 & 1 & 1 & 1 & 0  \\
\end{array}
\right)
\end{equation}
it is easy to check that $H \cdot (1, \dots, 1)^{T} = 0 \in \mathbb{F}_{2}^{12},$ so $(1, \dots, 1) \in \mathcal{C}_{2}$ which implies that $\mathcal{C}_{1} \subseteq \mathcal{S}_{2}(0,\dots,0).$ 

	Moreover, an element $c_2 \in \mathcal{C}_{2}$ can be written as $c_2=G.h,$ where $G=\begin{pmatrix}
I_{12 \times 12} \\
\hline 
B_{12 \times 12} 
\end{pmatrix}$ is the generator matrix of the Golay code and $h = (h_1, h_2, h_3, h_4, h_5, h_6, h_7, h_8 , h_9, h_{10}, h_{11}, h_{12})^{T} \in \mathbb{F}_{2}^{12}.$ Thus, when we sum all the coordinates of the resulting vector $c_{2} = G.h$ we have $8 h_1 +8 h_{2} + 8 h_{3} + 8 h_{4} + 8h_{5} + 8h_{6} + 8 h_{7} + 8 h_{8} + 8 h_{9}+8h_{10} +8h_{11} + 12h_{12} \equiv 0 \mod 2 \Rightarrow c_{2} \in \tilde{\mathcal{C}}_{3} = \mathcal{S}_{3}(0,\dots,0).$ Hence,
\begin{equation}
\mathcal{C}_{1} \subseteq \mathcal{S}_{2}(0,\dots,0) \subseteq \mathcal{C}_{2}  \subseteq  \mathcal{S}_{3}(0,\dots,0) \subseteq \mathcal{C}_{3}.
\end{equation}

	We still need to prove that 
\begin{itemize}
\item $\mathcal{S}_{2}(0,\dots,0)$ closes $\mathcal{C}_{1}$ under Schur product and this is clearly true because the Schur product of any elements in $\mathcal{C}_{1}$ belong to $\mathcal{S}_{2}(0,\dots,0).$

\item $\mathcal{S}_{3}(0,\dots,0)$ closes $\mathcal{C}_{2}$ under Schur product: if we consider $c_{2}= G.h \in \mathcal{C}_{2}$ and $\tilde{c}_{2}=G.\tilde{h} \in \mathcal{C}_{2},$ we have checked computationally that the sum of all coordinates of the Schur product $c_{2} \ast \tilde{c}_{2} \equiv 0 \ \mod 2$ $\Rightarrow c_{2} \ast \tilde{c}_{2} \in \mathcal{S}_{3}(0,\dots,0)= \tilde{C}_{3}.$ 

\end{itemize}

\end{example}

	We can have a lattice Construction $C^\star$ even when the nesting condition in Theorem \ref{thm} is not satisfied.

\begin{example} \label{exnonesting} Consider the linear binary code $\mathcal{C}=\{(0,0,0,0,0,0),(1,0,1,1,0,1),(0,0,1,0,1,1),(1,0,0,1,1,0),$ $ (0,0,0,0,1,0),(0,0,1,0,0,1),(1,0,0,1,0,0),(1,0,1,1,1,1)\}$ $\subseteq \mathbb{F}_{2}^{6}$ with $L=3, n=2.$ Observe that $\mathcal{C}_{1} \nsubseteq \mathcal{S}_{2}(0,0,0,0)$ and 
\begin{eqnarray}
\Gamma_{C^{\star}}=\{(0,0)+8z, (1,2)+8z, (2,4)+8z, (3,6)+8z, \nonumber \\
(4,0)+8z, (5,2)+8z, (6,4)+8z, (7,6)+8z\}
\end{eqnarray}
with $z \in \mathbb{Z}^{2},$ is a lattice.

\end{example} 

	To prove Theorem \ref{thm} we need to introduce the following auxiliary result:

\begin{lemma} (\textit{Sum in $\Gamma_{C^{\star}}$}) \label{lemmasum} Let $\mathcal{C} \subseteq \mathbb{F}_{2}^{nL}$ be a linear binary code. If $x,y \in \Gamma_{C^{\star}}$ are such that
\begin{eqnarray}
x&=&c_{1}+2c_{2}+\dots+2^{L-1}c_{L}+2^{L}z \label{eqx} \\
y&=&\tilde{c}_{1}+2\tilde{c}_{2}+\dots+2^{L-1}\tilde{c}_{L}+2^{L}\tilde{z}, \label{eqy}
\end{eqnarray}
with $(c_{1}, c_{2}, \dots, c_{L}), (\tilde{c}_{1}, \tilde{c}_{2}, \dots, \tilde{c}_{L}) \in \mathcal{C}$ and $z, \tilde{z} \in \mathbb{Z}^{n},$ then
{\small
\begin{eqnarray} \label{formulasum}
x+y &= &c_{1}\oplus \tilde{c}_{1} + 2(s_{1} \oplus (c_{2}\oplus \tilde{c}_{2})) + \dots + \nonumber \\
& & +2^{L-1}(s_{L-1} \oplus (c_{L}\oplus \tilde{c}_{L}))+ 2^{L}(s_{L}+z+\tilde{z}),
\end{eqnarray}}
where 
{\small \begin{eqnarray} \label{si}
& s_{i}= (c_{i} \ast \tilde{c}_{i}) \oplus r_{i}^{1} \oplus r_{i}^{2} \oplus \dots \oplus r_{i}^{i-1} = (c_{i} \ast \tilde{c}_{i}) \bigoplus\limits_{j=1}^{i-1} r_{i}^{j}, \nonumber \\
& r_{i}^{1}=(c_{i} \oplus \tilde{c}_{i}) \ast (c_{i-1} \ast \tilde{c}_{i-1}), \ \ r_{i}^{j}=r_{i}^{j-1} \ast r_{i-1}^{j-1}, \nonumber \\
& 2 \leq j \leq L-1, i=1, \dots, L. 
\end{eqnarray}}
\end{lemma}

\begin{proof} The proof is done by mathematical induction in the number of levels $L$ and it will be provided in the full paper \cite{bzc}.
\end{proof}

	The mathematical intuition behind Theorem \ref{thm} lies in the fact that since $a + b = a \oplus b + 2 (a \ast b)$  for $a,b \in \mathbb{F}_{2}^{n},$ when adding two points in $\Gamma_C$ or $\Gamma_{C^\star},$ each level $i \geq 2$ has the form of $c_{i} \oplus \tilde{c}_i \oplus carry_{(i-1)},$ where $carry_{(i-1)}$ is the "carry" term from the addition in the lower level. Since the projection code $\mathcal{C}_i$ is linear, $c_i \oplus \tilde{c}_i$ is a codeword in the $i-$th level. Hence, closeness of $\Gamma_{C^\star}$ under addition amounts to the fact that $carry_{(i-1)}$ is also a codeword in $\mathcal{C}_i,$ which is essentially the condition of the theorem. Formally, 
	

\begin{proof} [Proof of Theorem 2] $(\Leftarrow)$ 
 For any $x, y \in \Gamma_{C^{\star}},$ written as in Equations (\ref{eqx}) and (\ref{eqy}), we have $x+y$ as given in Lemma \ref{lemmasum} (Equations (\ref{formulasum}) and (\ref{si})) and we need to verify if $x+y \in \Gamma_{C^{\star}}.$ 

	Clearly $x+y \in \mathcal{C}_{1}+2\mathcal{C}_{2}+\dots+2^{L-1}\mathcal{C}_{L}+2^{L}\mathbb{Z}^{n}.$ It remains to demonstrate that $(c_{1}\oplus \tilde{c}_{1},  s_{1} \oplus c_{2}\oplus \tilde{c}_{2}, \dots, s_{L-1} \oplus c_{L} \oplus \tilde{c}_{L}) \in \mathcal{C}.$
	
	Indeed, using the fact that the chains  $\mathcal{C}_{i-1} \subseteq \mathcal{S}_{i}(0, \dots, 0)$ for all $i=2, \dots, L$ are closed under the Schur product, it is an element of $\mathcal{C}$ because it is a sum of elements in $\mathcal{C},$ i.e., 
\begin{eqnarray}
&& (c_{1}\oplus \tilde{c}_{1},  s_{1} \oplus c_{2}\oplus \tilde{c}_{2}, \dots, s_{L-1} \oplus c_{L}\oplus \tilde{c}_{L}) =  \nonumber \\
&&\underbrace{(c_{1}\oplus \tilde{c}_{1},c_{2}\oplus \tilde{c}_{2}, \dots,  c_{L}\oplus \tilde{c}_{L})}_{\in \mathcal{C}} \oplus  \underbrace{(0,s_{1}, \dots, 0)}_{\in \mathcal{C}} \oplus \dots \oplus \nonumber \\
&&\oplus \underbrace{(0,\dots, 0, s_{L-1})}_{\in \mathcal{C}}.
\end{eqnarray}
Observe that any $nL-$tuple $(0, \dots, s_{i-1}, \dots, 0)$ is in $\mathcal{C}$ because by hypothesis, the chain $\mathcal{S}_{i}(0, \dots, 0)$ closes $\mathcal{C}_{i-1}$ under Schur product, hence $S_{i}(0, \dots, 0)$ contains $(c_{i-1}*\tilde{c}_{i-1}), r_{i-1}^1,....,r_{i-1}^{i-2}$ which is sufficient to guarantee that $s_{i-1} \in \mathcal{S}_{i}(0,\dots,0)$ so $(0, \dots, s_{i-1}, \dots, 0) \in \mathcal{C},$ for all $i=2, \dots, L-1.$


\noindent $(\Rightarrow)$ For the converse, we know that $\Gamma_{C^\star}$ is a lattice, which implies that if $x,y \in \Gamma_{C^\star}$ then $x+y \in \Gamma_{C^\star}.$ From the notation and result from Lemma \ref{lemmasum}, more specifically Equations (\ref{eqx}), (\ref{eqy}), (\ref{formulasum}) and (\ref{si}), it means that 
\begin{equation}
(c_{1}\oplus \tilde{c}_{1}, s_{1} \oplus (c_{2}\oplus \tilde{c}_{2}), \dots,s_{L-1} \oplus (c_{L}\oplus \tilde{c}_{L})) \in \mathcal{C}.
\end{equation}

	We can write this $L-$tuple as
{\small\begin{eqnarray}
& & \underbrace{(c_{1}\oplus \tilde{c}_{1}, s_{1} \oplus (c_{2}\oplus \tilde{c}_{2}), \dots,s_{L-1} \oplus (c_{L}\oplus \tilde{c}_{L}))}_{\in \mathcal{C}}  = \nonumber \\ 
& & \underbrace{(c_{1}\oplus \tilde{c}_{1}, c_{2}\oplus \tilde{c}_{2}, \dots, c_{L}\oplus \tilde{c}_{L})}_{\in \mathcal{C}, \ \text{by linearity of} \ \mathcal{C}} \oplus (0, s_{1}, \dots, s_{L-1}) \\
& \Rightarrow & (0, s_{1}, \dots, s_{L-1}) \in \mathcal{C}. \label{eqshurp}
\end{eqnarray}}
Notice that we have 
\begin{eqnarray}
s_{1} & = & c_{1} \ast \tilde{c}_{1} \\
s_{2} & = &((c_{1} \ast \tilde{c}_{1}) \ast (c_{2} \oplus \tilde{c}_{2})) \oplus (c_{2} \ast \tilde{c}_{2}) \\
s_{3} & = & ((c_{3} \nonumber \oplus \tilde{c}_{3}) \ast (c_{2} \ast \tilde{c}_{2})) \ast (c_{2} \oplus \tilde{c}_{2} \ast (c_1 \ast \tilde{c}_{1})) \nonumber \\
&  & \oplus \ ((c_{3} \oplus \tilde{c}_{3}) \ast (c_{2} \ast \tilde{c}_{2})) \oplus (c_{3} \ast \tilde{c}_{3}) \\
& \vdots & \nonumber
%
\end{eqnarray}

	Due to the nesting $\mathcal{C}_{1} \subseteq \mathcal{S}_{2}(0,\dots,0) \subseteq \dots \subseteq \mathcal{C}_{L-1} \subseteq \mathcal{S}_{L}(0, \dots,0) \subseteq \mathcal{C}_{L},$ we can guarantee that there exist codewords whose particular Schur products $c_{i} \ast \tilde{c}_{i} = 0,$ for $i=1, \dots, L-2.$ Thus, 
\begin{equation}
s_{L-1} = (c_{L-1} \ast \tilde{c}_{L-1})
\end{equation}
and from Equation (\ref{eqshurp}), $(0,0, \dots, c_{L-1} \ast \tilde{c}_{L-1}) \in \mathcal{C},$ i.e., $S_{L}(0, \dots, 0)$ must close $\mathcal{C}_{L-1}$ under Schur product. Proceeding similarly, we demonstrate that $S_{i}(0, \dots, 0)$ must close $\mathcal{C}_{i-1},$ for all $i=2, \dots, L$ and it completes our proof.



\end{proof} 

\section{Conclusion and future work}

	In this paper a new method of constructing constellations was introduced, denoted by Construction $C^\star,$ which is subset of Construction C and is based on a modern coding scheme, the bit-interleaved coded modulation (BICM). It was proved when this construction is a lattice and how to describe the Leech lattice using this technique. 
	
	Our future work include examining on a comparative basis the advantages of Construction $C^{\star}$ compared to Construction C in terms of packing density. We also aim to change the natural labeling $\mu$ to the Gray map, the standard map used in BICM. Another direction to be completed is to find a more complete condition for the latticeness of Construction $C^{\star},$ that covers cases such as the one in Example \ref{exnonesting}.


\section*{Acknowledgment}

	CNPq (140797/2017-3, 312926/2013-8) and FAPESP (2013/25977-7) supported MFB and SIRC. Israel Science Foundation grant 676/15 supported RZ.


\begin{thebibliography}{1}


\bibitem{agrelleriksson} \label{agrelleriksoon}
E.~Agrell and T.~Eriksson, "Optimization of Lattice for Quantization". \emph{IEEE Trans. Inf. Theory}, vol. 44, no. 5, pp. 1814-1828, Sep. 1998.

\bibitem{amrani} \label{amrani}
O.~Amrani \emph{et al}, "The Leech Lattice and the Golay Code: Bounded-Distance Decoding and Mu1tilevel Constructions". \emph{IEEE Trans. Inf. Theory}, vol. 40, no. 4, pp. 1030-1043, Jul. 1994.

\bibitem{bollaufzamir} 
M.~F. Bollauf and R.~Zamir, "Uniformity properties of Construction C", in \emph{2016 IEEE Inter. Symp. on Inform. Theory}, (Barcelona), 2016, pp. 1516-1520.

\bibitem{bzc}
M.~F. Bollauf, R.~Zamir and S.I.R.~Costa, "Constructions C and $C^{\star}:$ theoretical and practical approach", \textit{in preparation}.

\bibitem{bonnecaze} \label{bonnecaze}
A.~Bonnecaze \emph{et al}, "Quaternary Quadratic Residue
Codes and Unimodular Lattices". \emph{IEEE Trans. on Inform. Theory}, vol. 41, no. 2, pp. 366-377, Mar. 1995.

\bibitem{caire}
G.~Caire, G.~Taricco and E.~Biglieri, "Bit-interleaved coded modulation". \emph{IEEE Trans. on Inform. Theory}, vol. 44, no. 3, pp. 927-946, May. 1998.

\bibitem{conwaysloane} \label{conwaysloane}
J.~H. Conway and N.J.~A. Sloane, \emph{Sphere Packings, Lattices and Groups}, 3rd~ed. New York, USA: Springer, 1999.
  
\bibitem{debuda1} \label{debuda1}
R.~de Buda, "Fast FSK signals and their demodulation". \emph{Can. Electron. Eng. Journal}, vol. 1, pp. 28–34, Jan. 1976.
  
\bibitem{forney1} \label{forney1}
G.~D. Forney, "Coset codes-part I: introduction and geometrical classification". \emph{IEEE Trans. Inf. Theory}, vol. 34, no. 5, pp. 1123-1151. Sep. 1988.

\bibitem{kositoggier} \label{kositoggier}
W.~Kositwattanarerk and F.~Oggier, "Connections between Construction D and related constructions of lattices". \emph{Designs, Codes and Cryptography}, v. 73, pp. 441-455, Nov. 2014. 

\bibitem{zehavi} \label{zehavi}
E.~Zehavi, "8-PSK trellis codes for a Rayleigh channel". \emph{IEEE Trans. Commun.}, vol. 40, no. 3, pp. 873–884, May 1992.

\end{thebibliography}
\end{document}